\PassOptionsToPackage{unicode}{hyperref}
\PassOptionsToPackage{hyphens}{url}
\PassOptionsToPackage{dvipsnames,svgnames,x11names}{xcolor}
\documentclass[
]{article}
\usepackage{xcolor}
\usepackage{amsmath,amssymb}
\usepackage{iftex}
\ifPDFTeX
  \usepackage[T1]{fontenc}
  \usepackage[utf8]{inputenc}
  \usepackage{textcomp} 
\else 
  \usepackage{unicode-math} 
  \defaultfontfeatures{Scale=MatchLowercase}
  \defaultfontfeatures[\rmfamily]{Ligatures=TeX,Scale=1}
\fi
\usepackage{lmodern}
\ifPDFTeX\else
\fi
\IfFileExists{upquote.sty}{\usepackage{upquote}}{}
\IfFileExists{microtype.sty}{
  \usepackage[]{microtype}
  \UseMicrotypeSet[protrusion]{basicmath} 
}{}
\makeatletter
\@ifundefined{KOMAClassName}{
  \IfFileExists{parskip.sty}{%
    \usepackage{parskip}
  }{
    \setlength{\parindent}{0pt}
    \setlength{\parskip}{6pt plus 2pt minus 1pt}}
}{
  \KOMAoptions{parskip=half}}
\makeatother
\makeatletter
\ifx\paragraph\undefined\else
  \let\oldparagraph\paragraph
  \renewcommand{\paragraph}{
    \@ifstar
      \xxxParagraphStar
      \xxxParagraphNoStar
  }
  \newcommand{\xxxParagraphStar}[1]{\oldparagraph*{#1}\mbox{}}
  \newcommand{\xxxParagraphNoStar}[1]{\oldparagraph{#1}\mbox{}}
\fi
\ifx\subparagraph\undefined\else
  \let\oldsubparagraph\subparagraph
  \renewcommand{\subparagraph}{
    \@ifstar
      \xxxSubParagraphStar
      \xxxSubParagraphNoStar
  }
  \newcommand{\xxxSubParagraphStar}[1]{\oldsubparagraph*{#1}\mbox{}}
  \newcommand{\xxxSubParagraphNoStar}[1]{\oldsubparagraph{#1}\mbox{}}
\fi
\makeatother

\usepackage{longtable,booktabs,array}
\usepackage{calc} 
\usepackage{etoolbox}
\makeatletter
\patchcmd\longtable{\par}{\if@noskipsec\mbox{}\fi\par}{}{}
\makeatother
\IfFileExists{footnotehyper.sty}{\usepackage{footnotehyper}}{\usepackage{footnote}}
\makesavenoteenv{longtable}
\usepackage{graphicx}
\makeatletter
\newsavebox\pandoc@box
\newcommand*\pandocbounded[1]{
  \sbox\pandoc@box{#1}%
  \Gscale@div\@tempa{\textheight}{\dimexpr\ht\pandoc@box+\dp\pandoc@box\relax}%
  \Gscale@div\@tempb{\linewidth}{\wd\pandoc@box}%
  \ifdim\@tempb\p@<\@tempa\p@\let\@tempa\@tempb\fi
  \ifdim\@tempa\p@<\p@\scalebox{\@tempa}{\usebox\pandoc@box}%
  \else\usebox{\pandoc@box}%
  \fi%
}
\def\fps@figure{htbp}
\makeatother

\NewDocumentCommand\citeproctext{}{}

\makeatletter
 \let\@cite@ofmt\@firstofone
 \def\@biblabel#1{}
 \def\@cite#1#2{{#1\if@tempswa , #2\fi}}
\makeatother
\newlength{\cslhangindent}
\newlength{\csllabelwidth}
\newenvironment{CSLReferences}[2] 
 {\begin{list}{}{%
  \setlength{\itemindent}{0pt}
  \setlength{\leftmargin}{0pt}
  \setlength{\parsep}{0pt}
  \ifodd #1
   \setlength{\leftmargin}{\cslhangindent}
   \setlength{\itemindent}{-1\cslhangindent}
  \fi
  \setlength{\itemsep}{#2\baselineskip}}}
 {\end{list}}
\usepackage{calc}

\usepackage{amsmath}
\usepackage{amsthm}
\makeatletter
\@ifpackageloaded{caption}{}{\usepackage{caption}}
\AtBeginDocument{%
\ifdefined\contentsname
  \renewcommand*\contentsname{Table of contents}
\else
  \newcommand\contentsname{Table of contents}
\fi
\ifdefined\listfigurename
  \renewcommand*\listfigurename{List of Figures}
\else
  \newcommand\listfigurename{List of Figures}
\fi
\ifdefined\listtablename
  \renewcommand*\listtablename{List of Tables}
\else
  \newcommand\listtablename{List of Tables}
\fi
\ifdefined\figurename
  \renewcommand*\figurename{Figure}
\else
  \newcommand\figurename{Figure}
\fi
\ifdefined\tablename
  \renewcommand*\tablename{Table}
\else
  \newcommand\tablename{Table}
\fi
}
\@ifpackageloaded{float}{}{\usepackage{float}}
\floatstyle{ruled}
\@ifundefined{c@chapter}{\newfloat{codelisting}{h}{lop}}{\newfloat{codelisting}{h}{lop}[chapter]}
\floatname{codelisting}{Listing}

\makeatother
\makeatletter
\@ifpackageloaded{caption}{}{\usepackage{caption}}
\@ifpackageloaded{subcaption}{}{\usepackage{subcaption}}
\makeatother
\usepackage{bookmark}
\IfFileExists{xurl.sty}{\usepackage{xurl}}{} 
\hypersetup{
  pdftitle={Conformal prediction without knowledge of labeled calibration data},
  pdfauthor={, and },
  colorlinks=true,
  linkcolor={blue},
  filecolor={Maroon},
  citecolor={Blue},
  urlcolor={Blue},
  pdfcreator={LaTeX via pandoc}}

\title{Conformal prediction without knowledge of labeled calibration
data}
\author{Jonas Flechsig\textsuperscript{1,*} \and Maximilian
Pilz\textsuperscript{1,2}}
\date{2025-09-12}
\begin{document}
\maketitle
\begin{abstract}
We extend the method of conformal prediction beyond the case relying on
labeled calibration data. Replacing the calibration scores by suitable
estimates, we identify conformity sets \(\mathcal{C}\) for
classification and regression models that rely on unlabeled calibration
data. Given a classification model with accuracy \(1-\beta\), we prove
that the conformity sets guarantee a coverage of \begin{equation*}
\mathbb{P}\left(Y\in\mathcal{C}\right)\geq 1-\alpha-\beta
\end{equation*} for an arbitrary parameter \(\alpha\in(0,1)\). The same
coverage guarantee also holds for regression models, if we replace the
accuracy by a similar exactness measure. Finally, we describe how to use
the theoretical results in practice.
\end{abstract}

\textsuperscript{1} Fraunhofer Institute for Industrial Mathematics\\
\textsuperscript{2} Nürnberg School of Health, Ohm University of Applied
Sciences Nuremberg

\textsuperscript{*} Correspondence:
\href{mailto:jonas.flechsig@itwm.fraunhofer.de}{Jonas Flechsig
\textless{}jonas.flechsig@itwm.fraunhofer.de\textgreater{}}

\newcommand{\PP}{\mathbb{P}}
\newcommand{\RR}{\mathbb{R}}
\newtheorem{example}{Example}
\newtheorem{theorem}{Theorem}
\newtheorem{lemma}{Lemma}
\newtheorem{corollary}{Corollary}
\newtheorem{definition}{Definition}
\newtheorem{remark}{Remark}

\newcommand{\JFcomment}[1]{\textcolor{blue}{\textbf{JF:} #1}}
\newcommand{\MPcomment}[1]{\textcolor{green}{\textbf{MP:} #1}}

\section{Introduction}\label{introduction}

Conformal prediction is a generic tool for uncertainty quantification in
machine learning methods. For regression problems, it provides
uncertainty intervals of confidence \(1-\alpha\) for the predicted value
for an arbitrary \(\alpha \in (0,1)\). For classification problems,
conformity sets are provided. Those are subsets of all possible classes
and for a specific classification task, the true class is covered with a
probability of at least \(1-\alpha\) by the conformity set.

The great feature of conformal prediction is its coverage guarantee
without any assumption on the data distribution (Angelopoulos and Bates
2022). Moreover, it is a model-agnostic method, this means that it can
be applied to any prediction model. Detailed theoretical background on
conformal prediction can be found in Angelopoulos, Barber, and Bates
(2025). An introduction to conformal prediction and its application in
Python is given by Molnar (2023). Many extensions of the standard
conformal prediction framework exist. For example, Tibshirani et al.
(2020) treat the question of a covariate shift in the data, Gibbs and
Candès (2021) introduce conformal prediction with coverage guarantees in
the adversarial sequence model, cross conformal prediction is introduced
and further investigated by Vovk (2015) and Vovk et al. (2018), and
localized conformal prediction is presented by Guan (2022).

Computing uncertainty measures by conformal prediction requires a
calibration set that is independent of the training set but whose true
outcomes are known. This set is used to identify how certain (or
\emph{conformal}) different predictions are and to define a boundary
(depending on the value of \(\alpha\)) from which on a predicted value
is added to the conformity interval or set, respectively. However, this
approach assumes that a labeled calibration set is available and this is
not guaranteed in practice. Data protection issues may make it necessary
to only provide the trained model without any possibility to access the
training or any other labeled data.

In this paper, we describe how conformal prediction can still be applied
if only an unlabeled test set is available. Concretely, the true labels
are replaced with labels predicted by the machine learning model. We
show that still a coverage guarantee can be given that depends on
\(\alpha\) and the model's accuracy in Section~\ref{sec-proof}. This
result is illustrated on a simple example in Section~\ref{sec-example}.
In Section~\ref{sec-discussion}, we briefly sum up the result and
describe how it can be used in practice.

\section{Conformal prediction with knowledge of labeled calibration
data}\label{conformal-prediction-with-knowledge-of-labeled-calibration-data}

Let us suppose we have a set of pairs \((X,Y)\) where \(X\) is an input
and \(Y\) is the corresponding output, called the \textit{label} of
\(X\), and let us assume that the pairs \((X,Y)\) are independently
identically distributed (i.i.d.) random variables. Such a relation is
used to determine a function \(\hat{f}:I\rightarrow T\) that maps \(X\)
from a space of inputs \(I\) to \(\hat{f}(X)\) in a target space \(T\).

In fact, given a set of independent \textit{training data}, the function
\(\hat{f}\) is chosen to minimize the distance between \(\hat{f}(X)\)
and \(Y\) for all pairs \((X,Y)\) in the training data. If an input
\(X\) indeed determines the output \(Y\) up to an error, the function
\(\hat{f}\) found on the training data should also perform well on
independent test data with the identical distribution as the training
data.

How well a function \(\hat{f}\) describes the relationship between input
and output is measured by uncertainty scores. In the approach of
\textit{conformal prediction}, we usually keep a set of pairs
\(\{(X_k,Y_k)\mid 1\leq k\leq n\}\) as \textit{calibration data}. The
calibration data are not used for training. On these data, conformal
prediction uses a \textit{score function} \(s(X,Y)\) to quantify the
uncertainty of the trained model with the convention that larger values
indicate worse agreement between the input and output.

\begin{example}
\label{ex:score_functions}
If $\hat{f}$ is a classification model, a common score function is 
\begin{equation}
\label{ex:score_functions_class}
s(X,Y)=1-p_{\hat{f},Y}(X)
\end{equation}
where $p_{\hat{f},Y}(X)$ is the probability that the model $\hat{f}$ assigns the class $Y$ to the input $X$.
If $\hat{f}$ is a regression model, a common score function is 
\begin{equation}
\label{ex:score_functions_reg}
s(X,Y)=d(Y,\hat{f}(X))
\end{equation}
with $d$ a metric on the target space $T$.
\end{example}

The score function yields a set of \textit{calibration scores}

\begin{equation}
\label{eq:def_s_k}
s_k = s(X_k,Y_k) 
\end{equation}

for \(1\leq k\leq n\). Ordering along these calibration scores induces
an order on the calibration data. The order clearly depends on the
calibration data \((X_k, Y_k)\) and moreover on the model \(\hat{f}\),
i.e.~on the training data. Given \(\alpha\in(0,1)\), let \(q\) be the
\(\frac{\lceil(n+1)(1-\alpha)\rceil}{n}\)-quantile of the calibration
scores and

\begin{equation}
\mathcal{C}(X_{\text{test}})=\{y\mid s(X_{\text{test}},y)\leq q\}.
\end{equation}

The following theorem is the main motivation for the study of conformal
prediction. Its proof can be found in Angelopoulos and Bates (2022)
(Theorem 1 and Appendix D) or Papadopoulos et al. (2002) (Proposition
1). More theoretical background is provided by Vovk, Gammerman, and
Saunders (1999).

\begin{theorem}[Conformal coverage guarantee]
\label{thm:conformal_coverage_guarantee}
Suppose $(X_k,Y_k)$ for $1\leq k\leq n$ and $(X_{\text{test}}, Y_{\text{test}})$ are i.i.d. random variables. The quantile $q$ and the set $\mathcal{C}(X_{test})$ are defined as above. Then the following holds: 
\begin{equation}
\mathbb{P}(Y_{test}\in\mathcal{C}(X_{test}))\geq 1-\alpha.
\end{equation}
\end{theorem}

\section{Conformal prediction without knowledge of labeled calibration
data}\label{sec-proof}

\subsection{Model assumption}\label{model-assumption}

We want to focus on the situation where we cannot use a labeled subset
of data for calibration but unlabeled data are available. For
considering classification and regression models, we want to use the
following measure of a model's predictive performance:

\begin{definition}
\label{def:exactness_measure}
Let $\hat{f}$ be a prediction model. 
Given $\tilde{\beta}\in[0,\infty)$, and $\beta\in[0,1)$, we call the model $\hat{f}$ $(\tilde{\beta},\beta)$-exact if 
\begin{equation}
\label{eq:exactness_measure}
\mathbb{P}(Y\in\{Z\in T\mid\vert s(X,Z)-s(X,\hat{f}(X))\vert\leq\tilde{\beta}\})\geq 1-\beta,
\end{equation}
where $Y$ is the correct output given the input $X$.
\end{definition}

\begin{corollary}
\label{ex:exactness_measure_classification}
Let $\hat{f}$ be a classification model and
$s$ the score function from Example \ref{ex:score_functions}. 
Then $\hat{f}$ is $(0, \beta)$-exact if and only if $\hat{f}$ has an accuracy of at least $1-\beta$. 
\end{corollary}
\begin{proof}
The following equation shows that given the score function from Example~\ref{ex:score_functions} 
the probability from \eqref{eq:exactness_measure} coincides with the definition of the accuracy. 
\begin{align*}
& \; \mathbb{P}\left(Y\in\lbrace Z\in T\mid\vert 1-p_{\hat{f},Z}(X)-\left(1-p_{\hat{f},\hat{f}(X)}(X)\right)\vert\leq0\rbrace\right)
\\
= & \; \mathbb{P}\left(Y\in\lbrace Z\in T\mid\vert p_{\hat{f},Z}(X)-p_{\hat{f},\hat{f}(X)}(X)\vert\leq0\rbrace\right)
\\
= & \; \mathbb{P}\left(Y\in\text{argmax}_{Z\in T}p_{\hat{f},Z}(X)\right).
\end{align*}
In the last step, we use the fact that the probability $p_{\hat{f},Z}(X)$ is maximal for $Z=\hat{f}(X)$ among all $Z\in T$.
The last line is the definition of the accuracy. 
Thus, the Corollary follows.
\end{proof}

\begin{corollary}
\label{ex:exactness_measure_regression}
Let $\hat{f}:I\rightarrow T$ with $T=\mathbb{R}^l$ be a regression model with median absolute error $MedAE$
and $s$ the score function from Example \ref{ex:score_functions}
with $d:\mathbb{R}^l\times\mathbb{R}^l\rightarrow\mathbb{R}, (x,y)\mapsto \sum_{i=1}^l |x_i-y_i |$. 
Then $\hat{f}$ is $(MedAE, 0.5)$-exact.
Note that the same assertion holds for all other quantiles of $d(Y,\hat{f}(X))$.
\end{corollary}
\begin{proof}
The proof follows from the definition of the median absolute error:
\begin{align*}
& \;\mathbb{P}(Y\in\lbrace Z\in T\mid\vert d(Y, \hat{f}(X)) - d(\hat{f}(X),\hat{f}(X))
\vert\leq MedAE \rbrace) \\
=& \;\mathbb{P}\left(Y\in\left\lbrace Z\in T\mid\sum_{i=1}^l\vert Y_i - \hat{f}(X)_i \vert\leq MedAE \right\rbrace\right) \\
\geq& \;0.5.
\end{align*}
\end{proof}

\subsection{Proof of coverage
guarantee}\label{proof-of-coverage-guarantee}

Now we want to prove a conformal coverage guarantee based on unlabeled
calibration data. To achieve this, let \(n\) be the number of elements
in the set of calibration data. For this \(n\), let us fix the space
\((I\times T)^n\times T\), and equip it with a probability measure
\(\mathbb{P}\).

Since we consider unlabeled calibration data, we include the calibration
set in the probability space. In contrast, Angelopoulos and Bates (2022)
restrict the probability space to the \((n+1)\)-st component of this
probability space. However, we will be able to rely on the results from
Angelopoulos and Bates (2022), if we fix (unknown) labels
\(y_1,...,y_n\) of the calibration data. As it is standard, we shorten
notation using \begin{equation*}
\{a\leq b\} = \{\omega\in(I\times T)^n\times T\mid a(\omega)\leq b(\omega)\}.
\end{equation*}

The calibration scores are estimated as follows. Let \begin{equation}
\label{eq:def_T_tilde_beta_X}
T_{\tilde{\beta}}(X) = \left\{y\in T\mid\vert s(X,\hat{f}(X))-s(X,y)\vert\leq\tilde{\beta}\right\}
\end{equation} and let \(\hat{Y}=\hat{Y}(X)\) be an arbitrary element
from the following set: \begin{equation}
\label{eq:def_Y_k_hat}
\text{argmax}_{y\in T_{\tilde{\beta}}(X)}s(X,y).
\end{equation} Moreover, for each \(1\leq k\leq n\), we use the
following abbreviation: \begin{equation}
\label{eq:def_s_k_hat}
\hat{s}_k = s(X_k,\hat{Y}_k).
\end{equation}

The following three lemmas will be needed to prove the main theorem of
this paper.

\begin{lemma}
\label{lem:s_k_hat_geq_s_k}
If $\hat{f}$ is $(\tilde{\beta},\beta)$-exact, 
then for random variables $(X,Y)$ i.i.d. to the training data of $\hat{f}$ we obtain 
\begin{equation}
\label{eq:score_function_estimation_claim}
\mathbb{P}\left(s\left(X,\hat{Y}\right)\geq s(X,Y)\right)\geq 1-\beta.
\end{equation}
\end{lemma}
\begin{proof}
Since $\hat{f}$ is $(\tilde{\beta},\beta)$-exact, 
$\mathbb{P}\left(\vert s(X,\hat{f}(X))-s(X,Y)\vert\leq\tilde{\beta}\right)\geq 1-\beta$.
In this subset, the definition of $\hat{Y}$ (see \eqref{eq:def_T_tilde_beta_X} and \eqref{eq:def_Y_k_hat}) in particular implies  
\begin{equation}
\label{eq:score_function_estimation}
s(X,\hat{Y})\geq s(X,Y).
\end{equation}
This proves the claim.
\end{proof}

Lemma \ref{lem:s_k_hat_geq_s_k} is the basis to apply the following
Lemma to the paired random variable \((s(X,\hat{Y}),s(X,Y))\) and its
representatives \((\hat{s}_k,s_k)\) for \(1\leq k\leq n\) that are based
on the calibration set.

\begin{lemma}
\label{lem:technical_lemma_2}
Let $(U,V)$ be a paired random variable such that $\mathbb{P}\left( U \geq V \right) \geq 1-\beta$.
Let $(U_1,V_1), \dots, (U_n,V_n)$ be random variables that are i.i.d. to $(U,V)$.
Then, for any $\alpha \in (0,1)$ and the empirical quantiles $\hat{q}_U(1-\alpha)$
and $\hat{q}_V(1-\alpha)$ it holds that
\begin{equation*}
\mathbb{P}\left( \hat{q}_U(1-\alpha) \geq \hat{q}_V(1-\alpha) \right) \geq 1-\beta.
\end{equation*}
\end{lemma}
\begin{proof}
Without lose of generality we can assume that $U_1\leq U_2\leq\dots\leq U_n$. 
If we consider the random variables $V_i$, there exists a bijection 
\begin{equation*}
\varphi:\{1,\dots,n\}\rightarrow\{1,\dots,n\}, k\mapsto i_k
\end{equation*}  
such that $V_{i_1}\leq V_{i_2}\leq\dots\leq V_{i_n}$.
Let $k^\ast$ be the index of the quantile $\hat{q}_U(1-\alpha)$.
This further implies that $V_{i_{k^\ast}}$ is the quantile $\hat{q}_V(1-\alpha)$.

\textit{Claim:} 
There exists an index $j\leq k^\ast$ such that $V_j\geq\hat{q}_V(1-\alpha)$

\textit{Proof of the claim:} 
If we assume that no such index $j$ exists, the set $\varphi^{-1}(\{1,\dots,k^\ast\})$ is disjoint from $\{k^\ast,\dots,n\}$. 
These sets have cardinalities $\vert\varphi^{-1}(\{1,\dots k^\ast\})\vert=k^\ast$ and $\vert\{k^\ast,\dots,n\}\vert=n-k^\ast+1$. 
The cardinalities of the disjoint sets sum up to $n+1$ but both are contained in a set with $n$ elements - a contradiction.   
This proves the claim.

If $U_j\geq V_j$ for an index $j$ that satisfies the condition from the claim, we obtain:
\begin{equation}
\hat{q}_U(1-\alpha)=U_{k^\ast}\geq U_j\geq V_j\geq\hat{q}_V(1-\alpha)
\end{equation}
Since $\mathbb{P}\left( U \geq V \right) \geq 1-\beta$, the event $U_j\geq V_j$ has probability at least $1-\beta$ 
and we obtain $\mathbb{P}\left( \hat{q}_U(1-\alpha) \geq \hat{q}_V(1-\alpha) \right) \geq 1-\beta$.
\end{proof}

\begin{lemma}
\label{lem:technical_lemma}
Let $A$ and $B$ be events in a given probability space with probability measure $\mathbb{P}$. If $\mathbb{P}\left(B\right)>0$, then $\mathbb{P}\left(A\mid B\right)\geq \frac{\mathbb{P}\left(A\right)-\mathbb{P}\left(B^C\right)}{\mathbb{P}\left(B\right)}$. 
\end{lemma}
\begin{proof}
Using the law of total probability, we observe: 
\begin{equation}
\label{eq:est_law_total_prob}
\mathbb{P}\left(A\mid B\right) = \frac{\mathbb{P}\left(A\right) - \mathbb{P}\left(A\mid B^C\right)\cdot \mathbb{P}\left(B^C\right)}{\mathbb{P}\left(B\right)} \geq \frac{\mathbb{P}\left(A\right) - \mathbb{P}\left(B^C\right)}{\mathbb{P}\left(B\right)}
\end{equation}
This proves the claim.
\end{proof}

\begin{theorem}
\label{thm:conformal_coverage_guarantee_without_labeled_data}
Suppose $(X_k,Y_k)$ for $1\leq k\leq n$ and $(X_{test}, Y_{test})$ are i.i.d. random variables, 
$\hat{f}$ is a $(\tilde{\beta},\beta)$-exact model and 
$\hat{q}=\frac{\lceil(n+1)(1-\alpha)\rceil}{n}$-quantile of $\hat{s}_k$ for $1\leq k\leq n$. 
If 
\begin{equation*}
\mathcal{C}(X_{test}) = \{((x_1,y_1),\dots,(x_n,y_n),y)\mid s\left(X_{test}, y\right)\leq \hat{q}(x_1,\dots,x_n)\}.
\end{equation*}
the following holds: 
\begin{equation}
\label{eq:ccg_without_labeled_data}
\mathbb{P}(((X_1,Y_1),\dots,(X_n,Y_n),Y_{test})\in\mathcal{C}(X_{test}))\geq 1-\alpha-\beta.
\end{equation}
\end{theorem}
\begin{proof}
Denote by $q$ the $\frac{\lceil(n+1)(1-\alpha)\rceil}{n}$-quantile of $s_k$
for $1\leq k\leq n$.

For better readability we use 
\begin{equation}
\label{eq:abb_set}
\{((X_1,Y_1),\dots,(X_n,Y_n),Y)\mid a(X_1,Y_1,\dots,X_n,Y_n,Y)\leq b(X_1,Y_1,\dots,X_n,Y_n,Y)\}
\end{equation}
to describe the event 
\begin{equation*}
((X_1,Y_1),\dots,(X_n,Y_n),Y)\in\{((x_1,y_1),\dots,(x_n,y_n),y)\mid a(x_1,y_1,\dots,x_n,y_n,y)\leq b(x_1,y_1,\dots,x_n,y_n,y)\}.
\end{equation*}
If we intersect two symbols of the form as in \eqref{eq:abb_set}, 
we mean that $((X_1,Y_1),\dots,(X_n,Y_n),Y)$ is contained in the intersection of the corresponding sets in the given probability space.

Let further
\begin{align*}
\hat{A} &= \{((X_1,Y_1),\dots,(X_n,Y_n),Y)\mid s\left(X_{test},Y\right)\leq \hat{q}(X_1,...,X_n)\},
\\
A &= \{((X_1,Y_1),\dots,(X_n,Y_n),Y)\mid s\left(X_{test},Y\right)\leq q(Y_1,...,Y_n)\} \text{ and } 
\\
B &= \{((X_1,Y_1),\dots,(X_n,Y_n),Y)\mid \hat{q}(X_1,\dots,X_n) \geq q(Y_1,\dots,Y_n)\}.
\end{align*}

Recall from Lemma \ref{lem:s_k_hat_geq_s_k} that $\mathbb{P}\left(s\left(X,\hat{Y}\right)\geq s(X,Y)\right)\geq 1-\beta$. 
This allows us to apply Lemma \ref{lem:technical_lemma_2} with $(U_k,V_k) = (\hat{s}_k, s_k)$ to obtain
$\mathbb{P}\left(B\right)\geq 1-\beta>0$.

By the definitions of $\hat{A}$, $A$, and $B$ it follows that $\mathbb{P}(\hat{A}\cap B)\geq\mathbb{P}(A\cap B)$.
Thus, we have the following lower bound:
\begin{equation*}
\mathbb{P}\left(\hat{A}\right) \geq \mathbb{P}\left(\hat{A}\cap B\right) 
\geq \mathbb{P}\left(A\cap B\right)
= \mathbb{P}\left(A\mid B\right) \cdot \mathbb{P}\left(B\right) 
\geq \mathbb{P}\left(A\right)-\mathbb{P}\left(B^C\right)
\end{equation*}
where the last estimation uses Lemma \ref{lem:technical_lemma}. 

The definition of $A$ and Theorem \ref{thm:conformal_coverage_guarantee} yield
\begin{equation*}
\mathbb{P}\left(A\right)
\geq 1-\alpha
\end{equation*}
Consequently, using the definition of $\hat{A}$ we obtain
\begin{equation*}
\mathbb{P}(\left((X_1,Y_1),\dots,(X_n,Y_n),Y_{test}\right)\in\mathcal{C}(X_{test})) = \mathbb{P}\left( \hat{A} \right) \geq 1-\alpha-\beta
\end{equation*}
which proves the theorem.
\end{proof}

\begin{remark}
\label{rem:estimate_can_be_sharp}
In \eqref{eq:ccg_without_labeled_data} equality can be approximately achieved. 
For example, if $\hat{f}$ is a classification model with accuracy $1$ and $\alpha\rightarrow0$, 
then $\mathbb{P}\left(\left((X_1,Y_1),\dots(X_n,Y_n),Y_{test}\right)\in\mathcal{C}\left(X_{test}\right)\right)\rightarrow1$.
\end{remark}

The following remark underlines that in practice we are interested in
choosing a minimal \(\tilde{\beta}\) such that \(\hat{f}\) is
\((\tilde{\beta},\beta)\)-exact.

\begin{remark}
\label{rem:choose_tilde_beta_minimal}
If $\tilde{\beta}_1<\tilde{\beta}_2$, 
then the definition in \eqref{eq:def_T_tilde_beta_X} yields $T_{\tilde{\beta}_1}(X_k)\subseteq T_{\tilde{\beta}_2}(X_k)$ for every $1\leq k\leq n$. 
This implies $\hat{s}_k(\tilde{\beta}_1)\leq\hat{s}_k(\tilde{\beta}_2)$ and consequently, 
$\hat{q}(\tilde{\beta}_1)\leq\hat{q}(\tilde{\beta}_2)$. 
Thus $\mathcal{C}(X_{test})$ for $\tilde{\beta}_1$ is contained in $\mathcal{C}(X_{test})$ for $\tilde{\beta}_2$, i.e. the function that maps $\tilde{\beta}$ to the corresponding set $\mathcal{C}(X_{test})$ is monotonically increasing.

Moreover, if $\tilde{\beta}\rightarrow\infty$, the set $T_{\tilde{\beta}}(X_k)$ converges to $T$ for every $1\leq k\leq n$. By \eqref{eq:def_Y_k_hat} this implies that $\hat{Y}_k$ is contained in $\text{argmax}_{y\in T} s(X_k, y)$. 

In the case where $\hat{f}$ is a regression model and $T$ is an unbounded metric space with $s$ as in \eqref{ex:score_functions_reg}, 
this implies $s\left(X_k,\hat{Y}_k\right)\rightarrow\infty$ for every $1\leq k\leq n$ and thus $\hat{q}\rightarrow\infty$. 
Hence, $\mathcal{C}(X_{test})=(I\times T)^n\times T$ and the estimation in \eqref{eq:score_function_estimation} is trivially satisfied. 

If in the classification case at least in $\lceil(n+1)\alpha\rceil$ of the $n$ calibration data, the classification probability $p_{\hat{f},\hat{Y}_k}(X_k)$ equals $0$, 
then as in the previous case $\mathcal{C}(X_{test})=(I\times T)^n\times T$ and the estimation in \eqref{eq:score_function_estimation} is trivially satisfied.
\end{remark}

\section{Illustration}\label{sec-example}

We present a simple example that demonstrates the result using the
statistical software R (R Core Team 2024). Let us assume a dataset with
the height, the weight, and the gender of 10000 persons (source: Arif
(2018)). To define a classification problem with more than two classes,
we denote a person with a BMI larger than 25 as ``overweighted'' and
``normal weighted'' otherwise following the official definition (World
Health Organisation 2000). The goal is to differ between
over-/normal-weighted women/men based on their height and weight. The
entire dataset is depicted in Figure~\ref{fig-data}.

\begin{figure}

\centering{

\pandocbounded{\includegraphics[keepaspectratio]{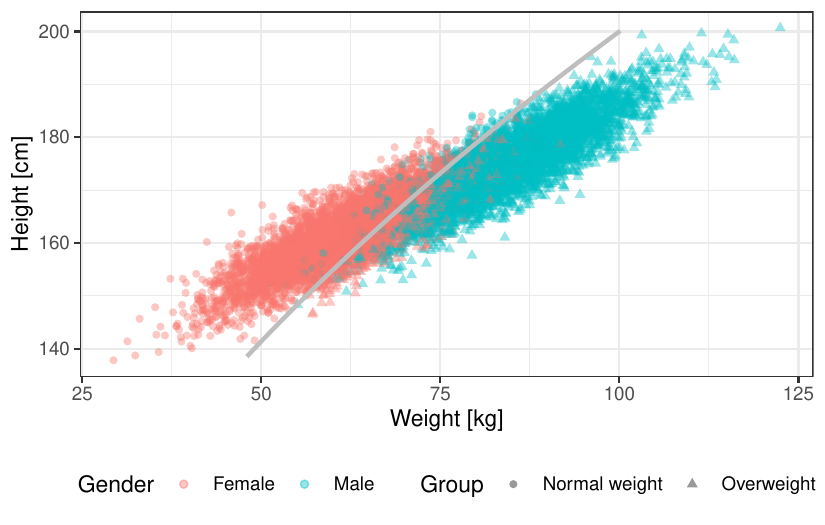}}

}

\caption{\label{fig-data}Dataset with height, weight, and classes. The
grey line indicates the limit for a BMI of 25.}

\end{figure}%

Now, the dataset is randomly split into three parts. The first part,
consisting of \(n_{train} = 6000\) persons is used to train a random
forest for the classification task described above. The second part with
\(n_{calib} = 1000\) has two purposes. First, it is used to estimate the
random forest's accuracy. Second, for the ``traditional'' conformal
prediction approach, it is used as a calibration set. Eventually, the
third dataset with \(n_{test} = 3000\) persons is used to evaluate how
well the conformal prediction approaches - using labeled and unlabeled
calibration data, respectively - work.

We compare the two conformal predictions approaches using a value of
\(\alpha = 0.05\).

For the first approach, we assume knowledge of the true outcomes in the
\(n_{calib}\) patients. The critical value for the predicted probability
of a class to be included in the conformity set is \(q = 0.174\) here.
Using this, we observe that for 95.1\% of the persons in the test
dataset, their conformity set includes their true class.

If we only use the calibration set to estimate the accuracy of the
random forest, we obtain an accuracy estimate of 88.7\%. Using the
predicted instead of the true values of the calibration data leads to an
inclusion boundary for the predicted probability of \(\hat{q} = 0.575\).
Consequently, for 86.8\% of the persons in the test dataset their
conformity set covers their true outcome. This is a larger value than
the value of 0.95 - 0.113 = 0.837, which is the lower bound for the
conformity according to the presented theoretical result.

Let us analyze one concrete person. This person from the test set is
female, 154cm tall and weighs 59kg. With an BMI of 24.9, she is a
``woman in normal weight'' but very close to being a ``woman with
overweight''. The random forest gives her a probability of 0.264 to be a
woman in normal weight and 0.625 to be a woman with overweight. Both
values are larger than \(q = 0.174\) such that the true class is in the
conformity set of the conformal prediction approach when the labels of
the calibration data are known. However, 0.264 is smaller than
\(\hat{q} = 0.575\) and therefore, the conformity set in case of unknown
labels of the calibration data only consists of the wrong class ``woman
with overweight''.

\section{Discussion}\label{sec-discussion}

In this paper, we presented a theoretical foundation of how conformal
prediction can be applied in situations where no labeled data (and thus
no calibration set) are available. The underlying definition of
\((\tilde{\beta}, \beta)\)-exactness may seem quite technically but can
be reduced to common measures as the accuracy for classification or
quantiles of the absolute prediction error in regression settings. Of
note, these quantities are not known in practice but can only be
estimated on the training data. This induces another source of
uncertainty and should be considered when interpreting conformal
prediction results derived by the presented approach.

It is an appealing property that the presented methodology can be
implemented model-agnostically and by only providing the trained model
and a performance number as accuracy or median absolute error. Given a
test set with sufficient number of subjects, the conformal prediction
intervals/sets can be derived by a leave-one-out approach. For each
observation, compute the quantile \(\hat{q}\) using all other test
observations and their predicted outcomes as described above. Then, the
conformal prediction interval for the missing observation can be derived
using the presented methodology. Rotating the observation that is not
used for computing \(\hat{q}\) gives a conformal prediction interval/set
for each subject in the test set as demonstrated in the example in
Section~\ref{sec-example}.

The present article introduces a foundational framework for conformal
prediction without knowledge of labeled calibration data. However,
variants that are more suited to the problems discussed in Tibshirani et
al. (2020), Gibbs and Candès (2021), Vovk (2015), Vovk et al. (2018),
and Guan (2022) in the case of unlabeled calibration data could be of
further interest. An essential question that emerges in these contexts
is a more suited formulation of \((\tilde{\beta},\beta)\)-exactness in
the specific settings. Exploring these settings in greater detail could
yield valuable insights.

\section*{References}\label{references}
\addcontentsline{toc}{section}{References}

\phantomsection\label{refs}

\end{document}